\documentclass[10pt,journal,onecolumn]{article}

\usepackage{my-details}

\usepackage[utf8]{inputenc} 
\usepackage[T1]{fontenc}    
\usepackage{hyperref}       
\hypersetup{
  colorlinks   = true, 
  urlcolor     = blue, 
  linkcolor    = blue, 
  citecolor   = blue 
}
\usepackage{amsfonts}       
\usepackage{microtype}      
\usepackage{graphicx}  
\usepackage[usenames]{color} 
\usepackage[UKenglish]{babel}
\usepackage{amssymb} 
\usepackage{amsmath} 
\usepackage{amsthm}
\theoremstyle{plain}
\newtheorem{theorem}{Theorem}[section]
\newtheorem{corollary}[theorem]{Corollary}
\newtheorem{lemma}[theorem]{Lemma}
\newtheorem{remark}[theorem]{Remark}
\newtheorem{definition}[theorem]{Definition}
\DeclareMathOperator{\Tr}{Tr}

\DeclareMathOperator{\id}{id}
\usepackage{caption}
\usepackage{physics}
\usepackage{mathtools}
\usepackage{bbm} 
\usepackage{float} 
\usepackage{xcolor} 

\date{}

\begin{document}
\title{Fully quantum arbitrarily varying channel coding for entanglement-assisted communication}

\author{
Paula Belzig,\\ Institute for Quantum Computing and Department of Applied Math, \\ University of Waterloo, Waterloo, ON, Canada, \\ Email: pbelzig@uwaterloo.ca
}

\maketitle

\begin{abstract}
If a sender and a receiver lack precise knowledge about the communication line that connects them, designing a scheme to reliably transmit information becomes more challenging. This has been studied in classical and quantum information theory in the context of compound channel models and arbitrarily varying channel models. However, a fully quantum version of system uncertainty allows for an even more challenging coding scenario with entangled channel uses. This type of model has previously been investigated for classical and quantum capacity. Here, we address the problem of entanglement-assisted capacity in the presence of such system uncertainty. We find that, under the assumption of a finite environment dimension, it is equal to a corresponding compound capacity. Intriguingly, our results imply that in certain fully quantum arbitrarily varying channel models, the entanglement-assisted capacity can be positive while the classical capacity is equal to zero, a phenomenon that does not occur in regular single-channel coding.
\end{abstract}

\section{Introduction}

\noindent Studying communication in the presence of noise in both classical and quantum systems has long been a topic of interest in mathematical theory and practical application. Linked by a fixed channel $T$ that captures the noise affecting the individual symbols in a message, the sender and the receiver aim to construct an encoding procedure that translates the message into a codeword which is robust to the noise and can still be decoded reliably. Then, the asymptotic rate of how many message bits can be transmitted per channel use with vanishing error using the best possible encoding and decoding procedure quantifies the channel's ability to transmit information. 

Communication naturally becomes more difficult when the sender and the receiver do not know the precise nature of the channel that links them. This uncertainty can be understood as a form of noise where the communication channel is susceptible to manipulation by the environment or even a malicious adversary. With only the promise that the channel is from a known family of channels, the goal is to design an encoder and a decoder that will work for any or most channels in the family. If it is assumed that the unkown channel will be the same for each channel use, this is referred to as coding for a \emph{compound channel}; if the channel may vary during each channel use, this is referred to as coding for an \emph{arbitrarily varying channel} (AVC). The intricacies of transmitting information over classical channels in the presence of such uncertainty have been studied in \cite{CK11,Wolfowitz78,BBT60}, where achievable rates and converse bounds under various assumptions have been established.

When the channels, sender and receiver are quantum, more notions of capacity become relevant. Important examples include the \emph{classical capacity} of a quantum channel \cite{Holevo96,SW97}, which measures the channel's effectiveness in transmitting classical information encoded in quantum states, the \emph{quantum capacity }\cite{Lloyd97,Shor02,Devetak03}, where the objective is to transmit quantum information itself through the channel, and the \emph{entanglement-assisted capacity} \cite{BSST99,BSST02}, where the sender and the receiver share quantum entanglement. It is well-known that access to entanglement gives rise to a single-letter capacity formula and leads to higher achievable communication rates in comparison with the unassisted case \cite{BW92,BSST99,BSST02}.



Quantum versions of compound channels and AVC have been introduced in \cite{AB07,ABBN12} for quantum channels from a known set. This can equivalently be viewed as a channel with an additional input state that dictates which channel from the set is used. Similarly to the classical case, this state can be viewed as the state of an environment, or the state of a third party (a \emph{quantum jammer}) perturbing the channel in order to hinder the communication.

However, the quantum nature of the environment also allows for a more general coding scenario with entangled environment states that correlate the channel uses with each other. The study of this, coding for a \emph{fully quantum arbitrarily varying channel} (FQAVC), has been initated in \cite{BDNW18} for the classical and quantum capacity. The entanglement-assisted case was left as an open problem, which we solve in this manuscript.

We find that, as long as the local dimension of this extra input state is finite, the entanglement-assisted capacity of a FQAVC  model reduces to a corresponding compound capacity. In other words, the entanglement-assisted capacity in the face of rather powerful and adversarial noise is, in fact, equal to a capacity under much more limited noise assumptions.

Interestingly, as a consequence of our result and the connection between entanglement and randomness \cite{MAG06,AMP12}, we show that there are FQAVC for which the entanglement-assisted capacity is positive while the classical capacity is zero. This is in stark contrast to the case of single-channel coding, where this is impossible for channels with finite input and output dimension \cite{BSST02,NC00}.


In Section~\ref{sec-definitions} of this manuscript, we give definitions of various notions of system uncertainty in classical and quantum communication. In Section~\ref{sec-coding-thms}, we give single-letter coding theorems for entanglement-assisted capacity in the various cases. In Section~\ref{sec-relations}, we discuss the relations between the various capacity models and point out an interesting consequence about how entanglement can boost the channel capacity.

\section{Communication setups with system uncertainty}
\label{sec-definitions}

\noindent In this section, we review the relevant notions 
and give definitions for the case of entanglement-assisted communication.

\subsection{Classical communication with system uncertainty}
\label{sec-classical}

\noindent The basic setup consists of a set of classical channels $\mathcal{T}=\{T_j\}$. The uncertainty in the coding setup manifests in a sender and a receiver that know this set, but do not know precisely which channel they are linked by.

Assuming that the sender and the receiver are linked by a specific channel $T_j$ from the set during each channel use (equivalently, by $n$ copies of the same channel, $T_j^{\otimes n}$), but not knowing which one, the goal is to design a coding scheme that works for any $j$ (and, in other words, for the worst of the channels from the set). This task is referred to as coding for the compound channel, or coding for a class of channels, and the best possible achievable rate in the asymptotic limit is referred to as the capacity of the compound channel generated by $\mathcal{T}$. If the sender and the receiver choose their codewords with a deterministic relationship to the message, this is referred to as deterministic coding, and we denote the best asymptotic achievable rate for a deterministic coding scheme for a compound channel model by $C_{comp}(\mathcal{T})$. Alternatively, if the sender and the receiver choose independently and uniformly at random from a set of codewords for each message, this is referred to as random coding, and the capacity for this scenario is denoted by $\bar{C}_{comp}(\mathcal{T})$.



On the other hand, the channel that links the sender and the receiver could vary for each channel use, meaning that they communicate via an unknown $n$-fold tensor product $ T_{j_1}\otimes T_{j_2} \otimes\cdots\otimes  T_{j_n}$ of channels from the set $\mathcal{T}$. Designing an encoder and a decoder for this situation is referred to as coding for the AVC generated by $\mathcal{T}$. Then, the corresponding capacities are denoted by ${C}_{AVC}(\mathcal{T})$ for the deterministic coding version and $\bar{C}_{AVC}(\mathcal{T})$ for the random coding version.


These concepts were introduced in order to study situations where the environment may interfere with the communication channel. For example, this could refer to a wire from a known material in an environment where the sender and the receiver do not know the precise temperature, which affects the wire's capacity. If the temperature stays the same but is unknown for the duration of the communication, this corresponds to a compound channel, and if the temperature may vary throughout the communication process, this corresponds to an AVC. However, beyond benign environmental influences, it could also be an adversary that controls the temperature in order to disrupt the communication. In reference to \cite{MS81,McEliece83}, such an adversary can be referred to as a \emph{jammer}. For the precise definitions and further exploration, we refer to \cite{CK11,Wolfowitz78,BBT60,McEliece83}.

As it turns out, many of these communication scenarios lead to the same best achievable communication rates in the asymptotic limit: For any set of channels $\mathcal{T}$, we have \[\bar{C}_{AVC}(\mathcal{T})=\bar{C}_{comp}(\mathcal{T})={C}_{comp}(\mathcal{T})\]
where the first equality is due to \cite{BBT60} and the second equality is due to \cite{Ahlswede78}.

Interestingly, one case differs: For AVC, random codes can generally outperform deterministic codes, and we have
    \begin{align*}
        {C}_{AVC}(\mathcal{T})= \begin{cases}
            \bar{C}_{AVC}(\mathcal{T}) & \text{ if } {C}_{AVC}(\mathcal{T})>0 \\
            0 & \text{otherwise}
        \end{cases}.
    \end{align*}

There exist channels for which ${C}_{AVC}(\mathcal{T})=0$ while $\bar{C}_{AVC}(\mathcal{T})>0$. An example of such a channel, based on \cite{BBT60}, is the binary adder channel described in \cite[Example~2]{CN88}. More generally, \cite{Ericson85} proved that a symmetrizability condition on the channel set implies ${C}_{AVC}(\mathcal{T})=0$, and \cite{CN88} proved that this symmetrizability is a necessary and sufficient condition. Furthermore, \cite[Appendix]{Ericson86} showed that there exist channels that fulfill this criterion while having $\bar{C}_{AVC}(\mathcal{T})>0$. A similar separation also manifests for quantum channels and has implications that we discuss in Section~\ref{sec-relations}.


\subsection{Uncertainty about the quantum channel}

\noindent A quantum channel is a completely positive and trace preserving map $T: \mathcal{M}_{d_A}\rightarrow \mathcal{M}_{d_B}$ where $\mathcal{M}_d$ denotes the matrix algebra of complex $d\times d$-matrices. For such channels, similar notions of system uncertainty can arise, and are explored in\cite{ABBN12,BCDN16,BJK17}.

\begin{definition}[Compound quantum channel]
Let $\mathcal{J}=\{1,...,J\}$. For $j\in \mathcal{J}$, let $T_j:\mathcal{M}_{d_A} \rightarrow \mathcal{M}_{d_B}$ be a quantum channel and let  $\mathcal{T}=\{T_j \}_{j}$ be a set of quantum channels. 
Then, the compound channel generated by $\mathcal{T}$ refers to the set \[\{T_j^{\otimes n} | j\in \mathcal{J},n\in\mathbbm{N}\}.\]
\end{definition}

\begin{definition}[{Arbitrarily varying quantum channel (AVQC)}]
Let $\mathcal{J}=\{1,...,J\}$. For $j\in \mathcal{J}$, let $T_j:\mathcal{M}_{d_A} \rightarrow \mathcal{M}_{d_B}$ be a quantum channel and let  $\mathcal{T}=\{T_j \}_{j}$ be a set of quantum channels. 
Then, the arbitrarily varying quantum channel generated by  $\mathcal{T}$ refers to the set
\[ \{ T_{j_1}\otimes T_{j_2} \otimes\cdots\otimes  T_{j_n} | j_k\in \mathcal{J} \forall k,n\in\mathbbm{N}\}.\]
\end{definition}

These models can be understood as a fixed quantum channel  $T:\mathcal{M}_{d_A}\otimes \mathcal{M}_{d_S} \rightarrow \mathcal{M}_{d_B}$ where the input to one subsystem, $\mathcal{M}_{d_S}$ with dimension $d_S$, is governed by some outside influence, be it an environment or a quantum jammer, who knows the code but not the codeword. The jammer can control each of the $n$ channel uses by choosing an input state in that subsystem. However, even more generally, in a quantum scenario, the jammer could insert an entangled input state $\sigma\in\mathcal{M}_{d_S}^{\otimes n}$ that introduces correlations between the $n$ separate channel uses. This version of system uncertainty was first studied in \cite{BDNW18}.



\begin{definition}[{Fully quantum arbitrarily varying channel (FQAVC)}]
Let $T:\mathcal{M}_{d_A}\otimes \mathcal{M}_{d_S} \rightarrow \mathcal{M}_{d_B}$ be a quantum channel and let $\sigma\in\mathcal{M}_{d_S}^{\otimes n}$ be a quantum state. Then, the fully quantum arbitrarily varying channel generated by $T$ refers to the channel 
$T^{\otimes n} (\cdot\otimes \sigma):\mathcal{M}_{d_A}^{\otimes n} \rightarrow \mathcal{M}_{d_B}^{\otimes n}$.
\end{definition}

\begin{remark}
    The compound channel and the AVQC are special cases of the FQAVC where the state $\sigma$ of the environment is restricted. Restricting to tensor power states $\sigma=\tau^{\otimes n}$, the FQAVC model is a compound channel with $T_\tau=T(\rho\otimes\tau)$. 
    If we consider tensor powers of varying states, the FQAVC model is an AVQC model.
\end{remark}

\subsection{Entanglement-assisted capacities with system uncertainty}

\noindent In this section, we define entanglement-assisted coding schemes for compound channels, AVQC and FQAVC. The first two were previously considered in \cite{BJK17,BGW17}. Corresponding schemes for classical communication have been introduced and analysed for compound channels in \cite{ABBN12}, for AVQC in \cite{AB07,ABBN12} and for FQAVC in \cite{BDNW18}. For quantum capacity, compound channels and AVQC have been studied in \cite{ABBN12}, and FQAVC were first considered in \cite{KMWY16}.




\begin{definition}[{Entanglement-assisted capacity of a compound channel}] \label{def-ea-comp}
    Consider the compound channel by the set $\mathcal{T}=\{T_j\}_j$ for $T_j:\mathcal{M}_{d_A} \rightarrow \mathcal{M}_{d_B}$ and $j\in \mathcal{J}$ with $\mathcal{J}=\{1,...,J\}$. Then, an $(n,m,k,\epsilon)$-coding scheme for entanglement-assisted communication over the compound channel consists of a pure quantum state $\Phi \in \mathcal{M}_{d_K}^{\otimes 2}$, a CPTP map $\mathcal{E}:\mathbbm{C}^{2^m} \otimes \mathcal{M}_{d_K} \rightarrow \mathcal{M}_{d_A}^{\otimes n}$ and a CPTP map\footnote{The respective definitions in \cite{BJK17} use an equivalent description of a coding scheme where the decoder is represented by a POVM instead of a quantum-classical channel.} $\mathcal{D}:\mathcal{M}_{d_B}^{\otimes n}\otimes \mathcal{M}_{d_K}\rightarrow \mathbbm{C}^{2^m}$, with 
    \[ \sup_{j} \frac{1}{M} \sum_{x=1}^{M} \trace\Big( (\mathbbm{1}-\mathcal{D}) \circ (( T_j^{\otimes n} \circ \mathcal{E})\otimes \id_k ) (\ketbra{x} \otimes \Phi)\Big) \leq \epsilon .\]

If, for some $k$ and for every $n\in \mathbbm{N}$, there exists an $(n,m(n),k,\epsilon(n))$-coding scheme for entanglement-assisted communication, then a rate $R\geq 0$ is called achievable for entanglement-assisted communication via the compound channel generated by $\mathcal{T}$ if $
   R  \leq 
   \liminf_{n\rightarrow \infty} \Big\{ 
    \frac{m(n)}{n} \Big\}
    $
and $\lim_{n\rightarrow \infty} \epsilon(n) \rightarrow 0$.
The entanglement-assisted compound capacity of $\mathcal{T}$ is given by
\begin{align*}
  C^{ea}_{comp}(\mathcal{T}) = \sup \{ R & |  R \text{ achievable for the compound channel} \\ &\text{ generated by }\mathcal{T} \}.
\end{align*}
\end{definition}


\begin{definition}[{Entanglement-assisted capacity of an AVQC}] \label{def-ea-avqc}
    Consider the AVQC generated by the set $\mathcal{T}=\{T_j\}_j$ for $T_j:\mathcal{M}_{d_A} \rightarrow \mathcal{M}_{d_B}$ and $j\in \mathcal{J}$ with $\mathcal{J}=\{1,...,J\}$. Then, an $(n,m,k,\epsilon)$-coding scheme for entanglement-assisted communication over the AVQC consists of a pure quantum state $\Phi$, a CPTP map $\mathcal{E}:\mathbbm{C}^{2^m} \otimes \mathcal{M}_{d_K} \rightarrow \mathcal{M}_{d_A}^{\otimes n}$ and a CPTP map $\mathcal{D}:\mathcal{M}_{d_B}^{\otimes n}\otimes \mathcal{M}_{d_K}\rightarrow \mathbbm{C}^{2^m}$, with coding error
    \[ \sup_{\{j_i\}_{i}} \frac{1}{M} \sum_{x=1}^{M} \trace\Big( (\mathbbm{1}-\mathcal{D}) \circ (( \bigotimes_{i=1}^n T_{j_i} \circ \mathcal{E})\otimes \id_k ) (\ketbra{x} \otimes \Phi)\Big) \leq \epsilon  \]
    where the supremum goes over all sequences $\{j_i\}_{i}$ for $i=1,...,n$ and $j_i\in \mathcal{J}$.
    
If, for some $k$ and for every $n\in \mathbbm{N}$, there exists an $(n,m(n),k,\epsilon(n))$-coding scheme for entanglement-assisted communication, then a rate $R\geq 0$ is called achievable for entanglement-assisted communication via the AVQC generated by $\mathcal{T}$ if $
   R  \leq 
   \liminf_{n\rightarrow \infty} \Big\{ 
    \frac{m(n)}{n} \Big\}
    $
and $\lim_{n\rightarrow \infty} \epsilon(n) \rightarrow 0$.
The entanglement-assisted AVQC capacity of $\mathcal{T}$ is given by
\begin{align*}
  C^{ea}_{AVQC}(\mathcal{T}) = \sup \{ R & |  R \text{ achievable for the AVQC} \\ &\text{ generated by }\mathcal{T} \}.
\end{align*}
\end{definition}

Based on the notions from \cite{ABBN12,BDNW18,CMH22}, we now introduce fully quantum versions of the communication setup. Since it is relevant for our proof strategy for Theorem~\ref{thm-FQAVC}, we differentiate between a capacity for deterministic coding and random coding\footnote{
While we have only stated definitions for deterministic coding here, Definiton \ref{def-ea-comp} and \ref{def-ea-avqc} similarly have random coding versions. As noted in \cite[Remark~1]{BJK17}, the asymptotic achievable rates for random coding are equal to the rates for deterministic coding for entanglement-assistance. We only define the random coding version for FQAVC since it plays a central role in the proof of our main result.

Further, we have chosen to consider communication under an error criterion averaged over the input messages. Alternatively, the averaged error criterion from Definitions \ref{def-ea-comp}, \ref{def-ea-avqc},\ref{def-ea-fqavc-d} and \ref{def-ea-fqavc-r} may be replaced with a maximal error criterion; \cite{BJK17} shows that this leads to the same asymptotic rates for entanglement-assisted communication for compound channels and for AVQC. For the FQAVC, the same holds true and can be seen by substituting the error criterion in the proof of Theorem~\ref{thm-FQAVC}.}. 

\begin{definition}[Entanglement-assisted capacity of an FQAVC, deterministic coding] \label{def-ea-fqavc-d}
    Consider the FQAVC generated by $T:\mathcal{M}_{d_A}\otimes \mathcal{M}_{d_S} \rightarrow \mathcal{M}_{d_B}$. Then, a deterministic $(n,m,k,\epsilon)$-coding scheme for entanglement-assisted communication over the FQAVC consists of a pure quantum state $\Phi$, a CPTP map $\mathcal{E}:\mathbbm{C}^{2^m} \otimes \mathcal{M}_{d_K} \rightarrow \mathcal{M}_{d_A}^{\otimes n}$ and a CPTP map $\mathcal{D}:\mathcal{M}_{d_B}^{\otimes n}\otimes \mathcal{M}_{d_K}\rightarrow \mathbbm{C}^{2^m}$, with coding error
    \begin{align} \label{eq-perr}
     &   p_{err}(T,\mathcal{E},\mathcal{D},\Phi) :=\sup_{\sigma} \frac{1}{M} \sum_{x=1}^{M} \trace\Big( (\mathbbm{1}-\mathcal{D}) \circ (( T^{\otimes n} \circ (\mathcal{E}\otimes \sigma) ) \otimes \id_k) (\ketbra{x} \otimes \Phi)\Big)
    \end{align} 
     with $ p_{err}(T,\mathcal{E},\mathcal{D},\Phi)\leq \epsilon$. 

If, for some $k$ and for every $n\in \mathbbm{N}$, there exists a deterministic $(n,m(n),k,\epsilon(n))$-coding scheme for entanglement-assisted communication, then a rate $R\geq 0$ is called achievable by deterministic coding for entanglement-assisted communication via the FQAVC generated by $T$ if $
   R  \leq 
   \liminf_{n\rightarrow \infty} \Big\{ 
    \frac{m(n)}{n} \Big\}
    $
and $\lim_{n\rightarrow \infty} \epsilon(n) \rightarrow 0$.
The entanglement-assisted FQAVC capacity of $T$ is given by
\begin{align*}
  C^{ea}_{FQAVC}(\mathcal{T}) = \sup \{ R & |  R \text{ achievable by deterministic coding} \\ &\text{ for the FQAVC generated by }T\}.
\end{align*}
\end{definition}

\begin{definition}[Entanglement-assisted capacity of an FQAVC, random coding] \label{def-ea-fqavc-r}
    Consider the FQAVC generated by $T:\mathcal{M}_{d_A}\otimes \mathcal{M}_{d_S} \rightarrow \mathcal{M}_{d_B}$. Then, a random $(n,m,k,\epsilon)$-coding scheme for entanglement-assisted communication over the FQAVC consists of a pure quantum state $\Phi$, CPTP maps $\mathcal{E}_y:\mathbbm{C}^{2^m} \otimes \mathcal{M}_{d_K} \rightarrow \mathcal{M}_{d_A}^{\otimes n}$ and CPTP maps $\mathcal{D}_y:\mathcal{M}_{d_B}^{\otimes n}\otimes \mathcal{M}_{d_K}\rightarrow \mathbbm{C}^{2^m}$ for $y=1,...,Y$, with coding error
        \[ \mathbbm{E}_y p_{err}(T,\mathcal{E}_y,\mathcal{D}_y,\Phi) \leq \epsilon\]
        with $p_{err}(T,\mathcal{E}_y,\mathcal{D}_y,\Phi)$ from Eq.~\eqref{eq-perr}.

If, for some $k$ and for every $n\in \mathbbm{N}$, there exists a random $(n,m(n),k,\epsilon(n))$-coding scheme for entanglement-assisted communication, then a rate $R\geq 0$ is called achievable by random coding for entanglement-assisted communication via the FQAVC generated by $T$ if $
   R  \leq 
   \liminf_{n\rightarrow \infty} \Big\{ 
    \frac{m(n)}{n} \Big\}
    $
and $\lim_{n\rightarrow \infty} \epsilon(n) \rightarrow 0$.
The entanglement-assisted random-coding FQAVC capacity of $T$ is given by
\begin{align*}
  \bar{C}^{ea}_{FQAVC}(T) = \sup \{ R & |  R \text{ achievable by random coding} \\ &\text{ for the FQAVC generated by }T \}.
\end{align*}
\end{definition}

\section{Coding theorems for entanglement-assisted communication with system uncertainty}
\label{sec-coding-thms}

\noindent In this section, we give coding theorems for entanglement-assisted communication under the various assumptions about system uncertainty, including known results from \cite{BGW17,BJK17} and our main result. 

\begin{theorem}[{Entanglement-assisted capacity of a compound channel, \cite[Theorem~1]{BJK17}}] \label{thm-ea-comp} For $j\in\{1,...,J\}$, let $T_j:\mathcal{M}_{d_A} \rightarrow \mathcal{M}_{d_B}$ be a quantum channel. Then, the entanglement assisted compound capacity of the set $\{T_j\}_j$ is given by
   \[ C^{ea}_{comp}(\{T_j\}_j)= \sup_{\rho_{A'A}} \inf_k I(A':B)_{(\id\otimes T_k)\rho_{A'A}}. \]
\end{theorem}

\begin{theorem}[{Entanglement-assisted capacity of a AVQC, \cite[Theorem~2]{BJK17}}]
For $j\in\{1,...,J\}$, let $T_j:\mathcal{M}_{d_A} \rightarrow \mathcal{M}_{d_B}$ be a quantum channel. Then, the entanglement assisted AVQC capacity of the set $\{T_j\}_j$ is given by
   \[ C^{ea}_{AVQC}(\{T_j\}_j)= \sup_{\rho_{A'A}} \inf_{T\in \text{conv}(T_{j_n})} I(A':B)_{(\id\otimes T)\rho_{A'A}}. \]
\end{theorem}

For a combination of a compound channel and a FQAVC with no assumption on the dimension, lower bounds on the entanglement-assisted capacity have been found in \cite{BCMH24}. However, these bounds may be quite suboptimal; in fact, for fully general FQAVC models, the lower bounds are trivial\footnote{In \cite{BCMH24}, the channel models contain an additional parameter $p$ that interpolates between a compound channel and a FQAVC. For any $0\leq p \leq 1$, any quantum channel $T:\mathcal{M}_{d_A}\rightarrow \mathcal{M}_{d_B}$, any quantum channel $N:\mathcal{M}_{d_A}\otimes \mathcal{M}_{d_S} \rightarrow \mathcal{M}_{d_B}$ and any quantum state $\sigma_S$, they define a channel $T_{p,N}=(1-p) (T\otimes \Tr_S) + pN $ and consider communication via $T_{p,N}^{\otimes n} (\cdot\otimes \sigma_S):\mathcal{M}_{d_A}^{\otimes n} \rightarrow \mathcal{M}_{d_B}^{\otimes n}$. This gives achievable rates which are lower bounded by a corresponding compound capacity reduced by some function $g(p)$. For $p\rightarrow 1$, which corresponds to the capacity of the FQAVC $N$, $g(1)$ always exceeds the compound capacity, and the lower bound is therefore zero.}. This is probably a consequence of the fact that, within their model, $d_S$ can be arbitrary and possibly dependent on $n$. Contrastingly, if the jammer's dimension is bounded with $d_S<\infty$ and not dependent on $n$, we can obtain positive achievable rates, and we even obtain equality to the corresponding compound capacity. This is the case that will be the focus of the remainder of this manuscript.

\begin{theorem} \label{thm-FQAVC}
Let $d_S<\infty$. The entanglement-assisted FQAVC capacity of the quantum channel $T:\mathcal{M}_{d_A}\otimes \mathcal{M}_{d_S} \rightarrow \mathcal{M}_{d_B}$ is given by
    \begin{align*}
        \bar{C}^{ea}_{FQAVC}(T)= C^{ea}_{comp}(\{T_{\sigma}\}_{\sigma}).
    \end{align*}
\end{theorem}
\begin{proof}
    Let $\mathcal{E}$, $\mathcal{D}$ be a $(n,m,k,\epsilon)$-coding scheme for entanglement-assisted compound channel coding. Due to \cite{BJK17}, we can find such a scheme\footnote{Note that the coding scheme for compound channels from \cite{BGW17} does not exhibit sufficiently fast decay of the coding error for our purposes here.} with $\epsilon=2^{-nc}$ for some positive constant $c$. Based on this encoder and decoder, we will construct a random $(n,m,k,\epsilon ')$-coding scheme for the FQAVC with $\epsilon'= (n+1)^{d_S^2}\epsilon$.
    
    Let $\pi$ be a permutation operator that acts on $\mathcal{M}_d^{\otimes n}$ by permuting the $n$ subsystems, and let $U_{\pi}$ be the associated quantum channel acting on a quantum state $\rho\in\mathcal{M}_d^{\otimes n}$ by $U_{\pi}(\rho)=\pi \rho \pi^{-1}$.
    
    Then, for any permutation of $n$ systems, consider the coding scheme given by an encoder $\mathcal{E}_{\pi}=U_{\pi} \circ \mathcal{E}$ and a decoder $\mathcal{D}_{\pi}=\mathcal{D}\circ U_{\pi^{-1}}$. This is a random coding scheme with coding error given by
    \begin{align*}
       \xi:= & \mathbbm{E}_{\pi} p_{err}(T,\mathcal{E}_{\pi},\mathcal{D}_{\pi},\Phi)\\& = \sup_{\sigma} \mathbbm{E}_{\pi} \frac{1}{M} \sum_{x=1}^{M} \trace\Big( (\mathbbm{1}-\mathcal{D}_{\pi})  \circ (( T^{\otimes n} \circ ( \mathcal{E}_{\pi} \otimes \sigma) ) \otimes \id_k) (\ketbra{x} \otimes \Phi)\Big) \\
        & =  \sup_{\sigma} \mathbbm{E}_{\pi}\frac{1}{M} \sum_{x=1}^{M} \trace\Big( (\mathbbm{1}-\mathcal{D}) \circ U_{\pi^{-1}} \circ  (( T^{\otimes n} \circ ( (U_{\pi} \circ  \mathcal{E})\otimes \sigma) ) \otimes \id_k) (\ketbra{x} \otimes \Phi)\Big) .
    \end{align*}
We insert an identity operation by permuting the subsystems of $\sigma$ by $\pi^{-1}$ and undoing the permutation by $\pi$, obtaining 
\begin{align*}
        \xi
        & =  \sup_{\sigma} \frac{1}{M} \sum_{x=1}^{M} \trace\Big( (\mathbbm{1}-\mathcal{D}) \circ (U_{\pi^{-1}} \circ T^{\otimes n} \circ U_{\pi} )  \circ ( \mathcal{E}\otimes \pi \sigma \pi^{-1}) ) \otimes \id_k) (\ketbra{x} \otimes \Phi)\Big) .
    \end{align*}
    Since the $n$ copies of the channel $T$ are permutation invariant, this is equal to
    \begin{align*}
       \xi 
       & =  \sup_{\sigma} \mathbbm{E}_{\pi} \frac{1}{M} \sum_{x=1}^{M} \trace\Big( (\mathbbm{1}-\mathcal{D}) \circ (T^{\otimes n} ) \circ ( \mathcal{E}\otimes \pi \sigma \pi^{-1}) ) \otimes \id_k) (\ketbra{x} \otimes \Phi)\Big) 
        \\& =  \sup_{\sigma} \frac{1}{M} \sum_{x=1}^{M} \trace\Big( (\mathbbm{1}-\mathcal{D}) \circ (T^{\otimes n} )  \circ  ( \mathcal{E}\otimes (\mathbbm{E}_{\pi} \pi \sigma \pi^{-1})) ) \otimes \id_k) (\ketbra{x} \otimes \Phi)\Big) .
    \end{align*}

    Due to the de Finetti theorem \cite{LW17}, any state $\rho\in\mathcal{M}_d^{\otimes n}$ that is invariant under permutations of the $n$ $d$-dimensional subsystems, i.e. $\pi \rho \pi^{-1} = \rho$, fulfills \[\rho \leq (n+1)^{d^2} \tau\]
    with the de Finetti state $\tau= \int  \sigma^{\otimes n} d(\sigma)$ where $d(\sigma)$ is a universal probability measure over the set of mixed states. Since $\mathbbm{E}_{\pi} \pi \sigma \pi^{-1}$ is permutation invariant, we can upper bound $\xi$ by an expression with the de Finetti state each term in the sum corresponds to the coding error of a compound channel, and the coding error for any compound channel is upper bounded by $\epsilon$ because of our choice of $\mathcal{E}$ and $\mathcal{D}$. Thereby, we have $\mathbbm{E}_{\pi}  p_{err}(T,\mathcal{E}_{\pi},\mathcal{D}_{\pi},\Phi)\leq (n+1)^{d_S^2}\epsilon =\epsilon'$.

    Then, since $\epsilon=2^{-nc}$ decays exponentially, as long as $d_S^2<\frac{cn}{\log(n+1)}$, we still have that $\lim_{n\rightarrow\infty}\epsilon'=\lim_{n\rightarrow\infty}(n+1)^{d_S^2} 2^{-nc} = 0$.
\end{proof}


\begin{corollary}
   As established by \cite[Remark~1]{BJK17}, an analogue of Theorem~\ref{thm-FQAVC} is true for entanglement-assisted capacities under deterministic coding, i.e. for any channel $T:\mathcal{M}_{d_A}\otimes \mathcal{M}_{d_S} \rightarrow \mathcal{M}_{d_B}$ (with $d_S<\infty$), we have $C^{ea}_{FQAVC}(T)= C^{ea}_{comp}(\{T_{\sigma}\}_{\sigma})$.
   %
\end{corollary}

\begin{remark}
Similarly to the case in \cite{BDNW18}, this proof strategy can not be employed without the restriction to $d_S<\infty$. While techniques from \cite{CMH22,BCMH24} can somewhat avoid this restriction, they only give a rather loose lower bound that does not lead to achievable rates under worst-case assumptions. For a detailed discussion of this dimension restriction, we refer to \cite[Section~IV]{BDNW18}.
\end{remark}

\section{About the power of entanglement-assistance}
\label{sec-relations}

\noindent In this section, we discuss the relationship between the entanglement-assisted capacity and its unassisted counterpart under system uncertainty assumptions. Based on definitions corresponding to Section~\ref{sec-definitions}, we denote by $C_{comp}(\mathcal{T})$ the compound classical capacity under deterministic coding generated by a set of channels $\mathcal{T}$ from \cite{AB07} and by $ C_{FQAVC}(T)$ the FQAVC classical capacity under deterministic coding of a quantum channel $T$ from \cite{BDNW18}.



\begin{lemma} \label{thm-comp-cap-relations} For $j\in\{1,...,J\}$, let $T_j:\mathcal{M}_{d_A} \rightarrow \mathcal{M}_{d_B}$ be a quantum channel. For any set $\mathcal{T}=\{T_j \}_{j}$, we have
  \[ C_{comp}(\mathcal{T})\leq C_{comp}^{ea}(\mathcal{T})\]
\end{lemma}

\begin{proof}
By \cite{BB09} and \cite[Corollary~3]{BDNW18}, we have have $C_{comp}(\mathcal{T})=\lim_{l\rightarrow \infty} \frac{1}{l} \sup_{p_x,\rho_x} \inf_j I(X:B^l)_{(\id\otimes T_j^{\otimes l})\rho}$.
Then, taking the supremum over a larger set of states can only increase this expression, and thus we have
$
    C_{comp}(\mathcal{T})  \leq \lim_{l\rightarrow \infty} \frac{1}{l} \sup_{\rho}  \inf_j  I(X:B^l)_{(\id\otimes T_j^{\otimes l})\rho}
$. 
Since this mutual information is additive under tensor powers of the channel, this is equal to the entanglement-assisted capacity of the compound channel, as given by Theorem~\ref{thm-ea-comp}.
\end{proof}

\begin{corollary} Let $d_S<\infty$. For any quantum channel $T:\mathcal{M}_{d_A}\otimes \mathcal{M}_{d_S} \rightarrow \mathcal{M}_{d_B}$, we have
  \[ {C}_{FQAVC}(T)\leq C_{FQAVC}^{ea}(T).\]
\end{corollary}

\begin{proof}
     Since ${C}_{FQAVC}(T)$ is either zero or equal to the random coding compound channel capacity, we have that ${C}_{FQAVC}(T)\leq {C}_{comp}(\{T_{\sigma}\})$ by \cite[Corollary~6]{BDNW18}. Furthermore, due to Theorem~\ref{thm-FQAVC}, we have $C_{FQAVC}^{ea}(T)=C_{comp}^{ea}(\{T_{\sigma}\})$. Then, Lemma~\ref{thm-comp-cap-relations} implies the statement.
\end{proof}



\begin{lemma} \label{thm-examples-exist-1}
    Let $d_S<\infty$. There exists a quantum channel $T:\mathcal{M}_{d_A}\otimes \mathcal{M}_{d_S} \rightarrow \mathcal{M}_{d_B}$ such that ${C}_{FQAVC}({T})=0$ while $C_{comp}(\{T_{\sigma}\}_{\sigma})>0$.
\end{lemma}

\begin{proof}
     For $T$ a classical channel, the examples discussed in Section~\ref{sec-classical} apply. Beyond classical channels, further examples can be found in \cite[Example~5.1]{BCDN16} or \cite[Section~8]{ABBN12}.
\end{proof}

\begin{lemma}
     Let $d_S<\infty$. There exists a quantum channel $T:\mathcal{M}_{d_A}\otimes \mathcal{M}_{d_S} \rightarrow \mathcal{M}_{d_B}$ such that ${C}_{FQAVC}(T)=0$ while ${C}^{ea}_{FQAVC}(T)>0$.
\end{lemma}

\begin{proof}
    This is a consequence of Lemma~\ref{thm-comp-cap-relations} and \ref{thm-examples-exist-1}.
\end{proof}

This highlights a notable aspect of the uncertain communication models we discuss in Section~\ref{sec-definitions}. For standard channel coding (for a single, fixed channel), the classical capacity is only equal to zero for quantum channels that also have entanglement-assisted capacity equal to zero \cite{NC00}. It has been an open question raised in one of the very first papers on entanglement-assisted communication whether entanglement between the sender and the receiver can lead to an unbounded improvement in the rates for finite-dimensional channels, and it is conjectured that the quotient $\frac{C^{ea}(T)}{C(T)}$ is bounded for any channel $T:\mathcal{M}_{d_A}\otimes \mathcal{M}_{d_S} \rightarrow \mathcal{M}_{d_B}$ by a function that depends only on the dimension $d_B$ \cite{BSST02}. For FQAVC and deterministic coding, a corresponding conjecture obviously does not hold, as the quotient $\frac{C_{FQAVC}^{ea}(T)}{C_{FQAVC}(T)}$ can diverge for some channels. 



\section{Conclusion and open questions}

\noindent In this manuscript, we propose a fully quantum version of entanglement-assisted communication with system uncertainty based on \cite{BJK17,BDNW18}. We demonstrate that, under a dimension restriction, the entanglement of a jammer state does not lead to a reduced asymptotic rate for entanglement-assisted communication. We show that it is still true (like in the single-channel scenario) that access to entanglement for the sender and the receiver increases the achievable rates, and it is even true (unlike in the single-channel scenario) that entanglement-assistance can increase it infinitely much.

It is intriguing to study under which assumptions entanglement between the sender and the receiver can and cannot increase achievable communication rates. A similar puzzling phenomenon to the one discussed in the previous section also manifests in \cite{Noetzel20}, where it is shown that entanglement-assistance can boost completely classical communication.

Furthermore, a better understanding of the gap between the bounds achievable for restricted and unrestricted jammer dimensions is a topic that warrants deeper investigation. The proof of our main result relies on offsetting the additional noise by choosing a coding scheme as a basis that has a coding error which decays sufficiently fast, and this strategy is employed (though in different ways) in both this work and to obtain the lower bound in \cite{BCMH24}. An improved lower bound, a proof that it is equal to the compound capacity, or a proof that communication is impossible for unrestricted jammers would be of considerable interest and shed light on the power of entanglement to disrupt communication instead of enhance it. 

\section*{Acknowledgement}

This research was undertaken thanks in part to funding from the Canada First Research Excellence Fund.

\bibliographystyle{marcotomPB}
 \bibliography{references} 

\end{document}